\documentclass[twoside,12pt]{article}
\usepackage{indentfirst}
\usepackage{bm}
\usepackage{graphicx}
\usepackage{epsfig}
\usepackage{amsmath}
\usepackage{amsfonts}
\usepackage{amssymb}
\usepackage{amsthm}
\usepackage{latexsym}

\newtheorem{thm}{Theorem}

\usepackage{epsfig}
\usepackage{amsmath}
\usepackage{epstopdf}
\usepackage{pgf,fancyhdr}
\usepackage{float}
\begin{document}

\begin{center}
{\Large\bf Projection based lower bounds of concurrence for multipartite quantum systems}
\end{center}

\begin{center}
\rm  Hui Zhao,$^1$ \  MeiMing Zhang,$^1$ \ Shao-Ming Fei$^{2,3}$  \ and Naihuan Jing,$^{4,5}$
\end{center}

\begin{center}
\begin{footnotesize} \sl
$^1$ College of Applied Sciences, Beijing University of Technology, Beijing 100124, China

$^2$ School of Mathematical Sciences, Capital Normal University, Beijing 100048, China

$^3$ Max-Planck-Institute for Mathematics in the Sciences, 04103 Leipzig, Germany

$^4$ Department of Mathematics, North Carolina State University, Raleigh, NC 27695, USA

$^5$ Department of Mathematics, Shanghai University, Shanghai 200444, China
\end{footnotesize}
\end{center}

\vspace*{2mm}

\begin{center}
\begin{minipage}{15.5cm} {\bf Abstract:}
\parindent 20pt\footnotesize
We study the concurrence of arbitrary-dimensional multipartite quantum states.
Analytical lower bounds of concurrence for tripartite quantum states are derived by projecting high-dimensional states to $2\otimes 2\otimes 2$ substates.
The results are then generalized to arbitrary multipartite quantum systems. Furthermore, the scheme enables us obtain lower bounds of
concurrence for arbitrary four-partite quantum states by projecting high-dimensional states to arbitrary given lower dimensional substates.
By detailed examples we show that our results improve the existing lower bounds of concurrence.
\end{minipage}
\end{center}

\begin{center}
\begin{minipage}{15.5cm}
\begin{minipage}[t]{2.3cm}{\bf Keywords:}\end{minipage}
\begin{minipage}[t]{13.1cm}
Concurrence, Entanglement, Monogamy inequality
\end{minipage}\par\vglue8pt
{\bf PACS: }03.65.Ud, 02.10.Ox, 03.67.Mn
\end{minipage}
\end{center}

\section{Introduction}
Quantum entanglement is a crucial feature of quantum mechanics. Entangled states are widely used in quantum information processing and quantum computation [1], such as quantum cryptographic schemes [2], entanglement swapping [3,4], quantum teleportation [5], dense coding [6] and so on.

The concurrence is one of the important measures of quantum entanglement. However, although concurrence is defined for arbitrary dimensional mixed quantum states, it is not easy to compute due to the extremum involved in the calculation. So far no explicit analytic formulae of concurrence have been found for systems larger than a pair of qubits [7], except for some special high dimensional bipartite symmetric states [8-11]. In terms of the substates and the generalized partial transposition criterion, analytical lower bounds of concurrence were presented for tripartite quantum systems [12-14]. Analytical lower bounds of concurrence for four-partite quantum states were provided in [15,16]. In Ref. [17], the authors derived a lower bound of concurrence for qubit quantum states.
The lower bounds of concurrence in terms of sub-states for tripartite quantum states were studied in [18], but the tripartite states were in the same dimensional systems. 
A generalized formula of concurrence for $n$-dimensional quantum systems was presented in [19]. Using the properties of the generalized concurrence, the entanglement of formation and the separability of high dimensional mixed states can be studied. An explicit lower bound of the concurrence for multipartite quantum states was derived in [20].
Considerable efforts have been devoted to multipartite concurrence [21,22].
Nevertheless, few analytic formulae for multipartite concurrence are known due to its complexity
compared with bipartite cases. 

In this paper, we study the lower bounds of concurrence for multipartite mixed quantum states. In Section 2, by projecting high-dimensional states to $2\otimes 2\otimes 2$ three-qubit substates and using the monogamy property of concurrence, we present an analytical lower bound of concurrence for any tripartite quantum state. The results are generalized to arbitrary multipartite systems. In Section 3, we project a high-dimensional four-partite quantum state to lower-dimensional ones and obtain a lower bound of concurrence for four-partite quantum states. By a detailed example we show that our results improve the existing lower bounds of concurrence. Comments and conclusions are given in Section 4.

\section{Lower bound of concurrence for multipartite systems from qubits substates} \label{The bound}
Let $H_i$, $i=1,2,\cdots,N$, be $d_i$-dimensional Hilbert spaces. The concurrence of an $N$-partite quantum pure state $|\varphi\rangle\in H_1\otimes\ H_2\otimes\cdots\otimes H_N$ is defined by [11],
\begin{eqnarray}
C_N(|\varphi\rangle)=2^{1-\frac{N}{2}}\sqrt{(2^N-2)-\sum_\alpha tr(\rho_\alpha^2)},
\end{eqnarray}
where the index $\alpha$ labels all $2^N-2$ subsystems of the $N$-partite quantum system and $\rho_\alpha$ are the reduced density matrices of $\rho=|\varphi\rangle\langle\varphi|$, $\rho_\alpha=tr_{\bar{\alpha}}(\rho)$, $\alpha\subset$ $\{1,2,\cdots,N$\}, $\bar{\alpha}$ is the compliment of $\alpha$.
For a mixed multipartite quantum state $\rho=\sum_ip_i|\varphi_i\rangle\langle\varphi_i|\in H_1\otimes\ H_2\otimes\cdots\otimes H_N$, the concurrence is given by the convex roof,
\begin{eqnarray}
C_N(\rho)=min_{\{p_i,|\varphi_i\rangle\}}\sum_ip_iC_N(|\varphi_i\rangle\langle\varphi_i|),
\end{eqnarray}
where the minimum is taken over all possible convex partitions of $\rho$ into pure state ensembles $\{p_i,|\varphi_i\rangle\}$, $0\leq p_i\leq1$ and $\sum_ip_i=1$.

For an $N$-partite quantum pure state $|\varphi\rangle \in H_1\otimes\ H_2\otimes\cdots\otimes H_N$, consider the general $M$-partite decomposition of $|\varphi\rangle$, $\{M_1,M_2,...,M_j\}$, with
$M_k$ partitions, $k=1,2,...,j$, each containing $k$ subspaces of $N$: $\sum_{k=1}^jM_k=M$, $\sum_{k=1}^jkM_k=N$.
The concurrence of the state $|\varphi\rangle$ under such $M$-partite partition is given by
\begin{eqnarray}
C_M(|\varphi\rangle)=2^{1-\frac{M}{2}}\sqrt{(2^M-2)-\sum_\beta tr(\rho_\beta^2)},
\end{eqnarray}
where $\beta\in\{M_1,M_2,...,M_j\}$.

We first consider the concurrence for tripartite quantum systems. A pure tripartite quantum state $|\varphi\rangle_{d_1\otimes d_2\otimes d_3}\in H_1\otimes H_2\otimes H_3$ with the dimensions $d_1$, $d_2$ and $d_3$, respectively, is of the form
\begin{eqnarray}
|\varphi\rangle_{d_1\otimes d_2\otimes d_3}=\sum_{i=1}^{d_1}\sum_{j=1}^{d_2}\sum_{k=1}^{d_3}a_{ijk}|ijk\rangle,
\end{eqnarray}
where $a_{ijk}\in$ $\mathbb{C}$, $\sum_{ijk}a_{ijk}$$a_{ijk}^\ast$=1. The squared concurrence of $|\varphi\rangle_{d_1\otimes d_2\otimes d_3}$, $C_3^2(|\varphi\rangle_{d_1\otimes d_2\otimes d_3})$, is given by [14],
\begin{eqnarray}
\begin{split}
C_3^2(|\varphi\rangle_{d_1\otimes d_2\otimes d_3})=\frac{1}{2}\sum_{i,p=1}^{d_1}\sum_{j,q=1}^{d_2}\sum_{k,t=1}^{d_3}(&|a_{ijk}a_{pqt}-a_{ijt}a_{pqk}|+|a_{ijk}a_{pqt}-a_{iqk}a_{pjt}|\\
&+|a_{ijk}a_{pqt}-a_{pjk}a_{iqt}|).
\end{split}
\end{eqnarray}
We project the tripartite state $|\varphi\rangle_{d_1\otimes d_2\otimes d_3}$ to (three-qubit) substates $|\varphi\rangle_{2\otimes 2\otimes 2}$ given by
\begin{eqnarray}
|\varphi\rangle_{2\otimes 2\otimes 2}=\sum_{i\in\{i_1,i_2\}}\sum_{j\in\{j_1,j_2\}}\sum_{k\in\{k_1,k_2\}}a_{ijk}|ijk\rangle,
\end{eqnarray}
where $i_1\neq i_2\in 1,...,d_1$, $j_1\neq j_2\in 1,...,d_2$, $k_1\neq k_2\in 1,...,d_3$.
There are $d_1\choose 2$$d_2\choose 2$$d_3\choose 2$ $=$\\$\frac{d_1d_2d_3(d_1-1)(d_2-1)(d_3-1)}{8}$ different substates.
By using Eq.(5) we have
\begin{eqnarray}\label{c327}
\nonumber C_3^2(|\varphi\rangle_{d_1\otimes d_2\otimes d_3})&=&\frac{d_1^2d_2^2d_3^2}{8d_1d_2d_3(d_1-1)(d_2-1)(d_3-1)}\sum C_3^2(|\varphi\rangle_{2\otimes 2\otimes 2})\\&\geq&\frac{1}{(d_1-1)(d_2-1)(d_3-1)}\sum C_3^2(|\varphi\rangle_{2\otimes 2\otimes 2}).
\end{eqnarray}
For a mixed state $\rho_{d_1\otimes d_2\otimes d_3}$, the corresponding three-qubit substates $\rho_{2\otimes2\otimes2}$ have the following form,
\begin{eqnarray}
\rho_{2\otimes2\otimes2}=
\left[ \begin{array}{cccccccc}
           \rho_{i_1j_1k_1},_{i_1j_1k_1}&\rho_{i_1j_1k_1},_{i_1j_1k_2}&\cdots&\rho_{i_1j_1k_1},_{i_2j_2k_1}&\rho_{i_1j_1k_1},_{i_2j_2k_2}  \\
           \rho_{i_1j_1k_2},_{i_1j_1k_1}&\rho_{i_1j_1k_2},_{i_1j_1k_2}&\cdots&\rho_{i_1j_1k_2},_{i_2j_2k_1}&\rho_{i_1j_1k_2},_{i_2j_2k_2}  \\
           \rho_{i_1j_2k_1},_{i_1j_1k_1}&\rho_{i_1j_2k_1},_{i_1j_1k_2}&\cdots&\rho_{i_1j_2k_1},_{i_2j_2k_1}&\rho_{i_1j_2k_1},_{i_2j_2k_2}  \\
           \vdots&\vdots&\vdots&\vdots&\vdots   \\
           \rho_{i_2j_1k_2},_{i_1j_1k_1}&\rho_{i_2j_1k_2},_{i_1j_1k_2}&\cdots&\rho_{i_2j_1k_2},_{i_2j_2k_1}&\rho_{i_2j_1k_2},_{i_2j_2k_2}  \\
           \rho_{i_2j_2k_1},_{i_1j_1k_1}&\rho_{i_2j_2k_1},_{i_1j_1k_2}&\cdots&\rho_{i_2j_2k_1},_{i_2j_2k_1}&\rho_{i_2j_2k_1},_{i_2j_2k_2}  \\
           \rho_{i_2j_2k_2},_{i_1j_1k_1}&\rho_{i_2j_2k_2},_{i_1j_1k_2}&\cdots&\rho_{i_2j_2k_2},_{i_2j_2k_1}&\rho_{i_2j_2k_2},_{i_2j_2k_2}
           \end{array}
      \right ],
\label{A}
\end{eqnarray}
which are unnormalized mixed ones.

A lower bound of concurrence for tripartite quantum states is given by the following result. 
\begin{thm}
For any $d_1\otimes d_2\otimes d_3$ tripartite quantum mixed state $\rho_{d_1\otimes d_2\otimes d_3}$, the concurrence $C_3(\rho_{d_1\otimes d_2\otimes d_3})$ satisfies
\begin{eqnarray}
C_3(\rho_{d_1\otimes d_2\otimes d_3})\geq\frac{1}{\sqrt{(d_1-1)(d_2-1)(d_3-1)}}[\sum \sum_{m=1}^2\sum_{n>m}^3C_{mn}^2(\rho_{2\otimes 2\otimes 2})]^\frac{1}{2},
\end{eqnarray}
where $\sum$ sums over 
all possible $2\otimes 2\otimes 2$ mixed substates $\rho_{2\otimes 2\otimes 2}$.
\end{thm}
\begin{proof}
For a pure quantum state $|\varphi\rangle_{2\otimes 2\otimes 2}$, according to (1) and (3) one has
{\setlength\abovedisplayskip{1pt}
\setlength\belowdisplayskip{1pt}
\begin{eqnarray}
C_3^2(|\varphi\rangle_{2\otimes 2\otimes 2})=\frac{1}{2}(C_{1|23}^2(|\varphi\rangle_{2\otimes 2\otimes 2})+C_{2|13}^2(|\varphi\rangle_{2\otimes 2\otimes 2})+C_{3|12}^2(|\varphi\rangle_{2\otimes 2\otimes 2})).
\end{eqnarray}}
The monogamy relation of concurrence [14] implies that 
{\setlength\abovedisplayskip{1pt}
\setlength\belowdisplayskip{1pt}
\begin{eqnarray}
C_{i|jk}^2(|\varphi\rangle_{2\otimes 2\otimes 2})\geq C_{ij}^2(|\varphi\rangle_{2\otimes 2\otimes 2})+C_{ik}^2(|\varphi\rangle_{2\otimes 2\otimes 2}), \ i\neq j\neq k=1,2,3.
\end{eqnarray}}
Therefore,
{\setlength\abovedisplayskip{1pt}
\setlength\belowdisplayskip{1pt}
\begin{eqnarray}
C_3^2(|\varphi\rangle_{2\otimes 2\otimes 2})\geq C_{12}^2(|\varphi\rangle_{2\otimes 2\otimes 2})+C_{13}^2(|\varphi\rangle_{2\otimes 2\otimes 2})+C_{23}^2(|\varphi\rangle_{2\otimes 2\otimes 2}).
\end{eqnarray}}
From (7) and (12), we obtain
\begin{eqnarray}
\begin{split}
C_3^2(|\varphi\rangle_{d_1\otimes d_2\otimes d_3})&\geq\frac{1}{(d_1-1)(d_2-1)(d_3-1)}\sum
C_3^2(|\varphi\rangle_{2\otimes 2\otimes 2})\\
&\geq\frac{1}{(d_1-1)(d_2-1)(d_3-1)}\sum\sum_{m=1}^2\sum_{n>m}^3C_{mn}^2(|\varphi\rangle_{2\otimes 2\otimes 2}).
\end{split}
\end{eqnarray}
For a mixed state $\rho_{d_1\otimes d_2\otimes d_3}=\Sigma_i p_i|\varphi_i\rangle\langle\varphi_i|$, we have
\begin{eqnarray}
  \begin{split}
C_3(\rho_{d_1\otimes d_2\otimes d_3})&=min \sum_ip_iC_3(|\varphi_i\rangle_{d_1\otimes d_2\otimes d_3})\\
&\geq min \frac{1}{\sqrt{(d_1-1)(d_2-1)(d_3-1)}}\sum_ip_i(\sum \sum_{m=1}^2\sum_{n>m}^3C_{mn}^2(|\varphi_i\rangle_{2\otimes 2\otimes 2}))^\frac{1}{2}\\
&\geq min
\frac{1}{\sqrt{(d_1-1)(d_2-1)(d_3-1)}}[\sum(\sum_ip_i\sum_{m=1}^2\sum_{n>m}^3C_{mn}^2(|\varphi_i\rangle_{2\otimes 2\otimes 2}))^2]^\frac{1}{2}    \\
&\geq \frac{1}{\sqrt{(d_1-1)(d_2-1)(d_3-1)}}[\sum\sum_{m=1}^2\sum_{n>m}^3(min \sum_ip_iC_{mn}^2(|\varphi_i\rangle_{2\otimes 2\otimes 2}))^2]^\frac{1}{2}    \\
&=\frac{1}{\sqrt{(d_1-1)(d_2-1)(d_3-1)}}[\sum \sum_{m=1}^2\sum_{n>m}^3C_{mn}^2(\rho_{2\otimes 2\otimes 2})]^\frac{1}{2},                     \\
  \end{split}
\end{eqnarray}
where we have used the Minkowski inequality $(\sum_j(\sum_ix_{ij})^2)^\frac{1}{2}\leq\sum_i(\sum_j x_{ij}^2)^\frac{1}{2}$ in the second inequality, the minimum is taken over all possible pure state decompositions of the mixed state $\rho_{d_1\otimes d_2\otimes d_3}$ in the first three minimizations, while the minimum in the last inequality is taken over all pure state decompositions of $\rho_{2\otimes 2\otimes 2}$.
\end{proof}
{\bf Remark 1.} Theorem 1 in [18] gives the lower bound of quantum states in a tripartite quantum system with subsystem dimensions N, respectively. We study the lower bound of concurrence of tripartite quantum states in a quantum system with different dimensions. And when m = 2, the lower bound of concurrence in our Theorem 1 is smaller than Theorem 1 in [18]. Theorem 2 in [16], the authors derive the bound of concurrence for $2\otimes 2\otimes 4$ quantum states, i.e, $C_3(\rho_{2\otimes 2\otimes 4})\geq\frac{1}{3}\sum C_3^2(\rho_{2\otimes2\otimes2})$, thus our Theorem 1 is a generalization of the Theorem 2 given in [16].

Next we consider the lower bound of concurrence for four-partite quantum systems. We first consider a pure bipartite quantum state $|\varphi\rangle_{d_1\otimes d_2}=\sum_{i=1}^{d_1}\sum_{j=1}^{d_2}a_{ij}|ij\rangle$, where $a_{ij}\in$ $\mathbb{C}$, $\sum_{ij}a_{ij}a_{ij}^\ast\\=1$. The concurrence of $|\varphi\rangle_{d_1\otimes d_2}$ can be written as $C^2(|\varphi\rangle_{d_1\otimes d_2})=\sum_{i,p=1}^{d_1}\sum_{j,q=1}^{d_2}|a_{ij}a_{pq}-a_{pj}a_{iq}|^2$.
The projected two-qubit substates of $|\varphi\rangle_{d_1\otimes d_2}$ are of the form,
$|\varphi\rangle_{2\otimes 2}$ $=$ $\sum_{i\in\{i_1,i_2\}}\sum_{j\in\{j_1,j_2\}}$\\$a_{ij}|ij\rangle$, where
$i_1\neq i_2\in 1,...,d_1$ and $j_1\neq j_2\in 1,...,d_2$.
For a pure bipartite quantum state $|\varphi\rangle_{d_1\otimes d_2}$, similar to (\ref{c327}), we have
\begin{eqnarray}
C^2(|\varphi\rangle_{d_1\otimes d_2})\geq\frac{1}{(d_1-1)(d_2-1)}\sum C^2(|\varphi\rangle_{2\otimes 2}),
\end{eqnarray}
where 
the summation runs over all possible $2\otimes 2$ pure sub-states $|\varphi\rangle_{2\otimes 2}$.

For any mixed quantum state $\rho\in H_1\otimes H_2\otimes H_3\otimes H_4$, the concurrence is bounded by [16]
{\setlength\abovedisplayskip{1pt}
\setlength\belowdisplayskip{1pt}
\begin{eqnarray}
 \begin{split} C_4^2(\rho)&\geq\frac{1}{12}(2C_{1|2|34}^2(\rho)+2C_{1|3|24}^2(\rho)+2C_{1|4|23}^2(\rho)+2C_{12|3|4}^2(\rho)+2C_{13|2|4}^2(\rho)\\
& +2C_{14|2|3}^2(\rho)+C_{12|34}^2(\rho)+C_{13|24}^2(\rho)+C_{14|23}^2(\rho)).
 \end{split}
\end{eqnarray}}
From (7), (15) and (16), we have the following theorem:
\begin{thm}
For any $d_1\otimes d_2\otimes d_3\otimes d_4$ mixed quantum state $\rho_{d_1\otimes d_2\otimes d_3\otimes d_4}$, the concurrence $C_4(\rho_{d_1\otimes d_2\otimes d_3\otimes d_4})$ satisfies
\begin{eqnarray}
 \begin{split}
 &C_4^2(\rho_{d_1\otimes d_2\otimes d_3\otimes d_4})\geq\frac{1}{12}(\sum_{\rho_{1|2|34}}\frac{2}{(d_1-1)(d_2-1)(d_3+d_4-1)}C_3^2(\rho_{2\otimes 2\otimes 2})\\
 &+\sum_{\rho_{1|3|24}}\frac{2}{(d_1-1)(d_3-1)(d_2+d_4-1)}C_3^2(\rho_{2\otimes 2\otimes 2})
 +\sum_{\rho_{1|4|23}}\frac{2}{(d_1-1)(d_4-1)(d_2+d_3-1)}\\
 &C_3^2(\rho_{2\otimes 2\otimes 2})+\sum_{\rho_{12|3|4}}\frac{2}{(d_3-1)(d_4-1)(d_1+d_2-1)}C_3^2(\rho_{2\otimes 2\otimes 2})\\
 &+\sum_{\rho_{13|2|4}}\frac{2}{(d_2-1)(d_4-1)(d_1+d_3-1)}C_3^2(\rho_{2\otimes 2\otimes 2})
 +\sum_{\rho_{14|2|3}}\frac{2}{(d_2-1)(d_3-1)(d_1+d_4-1)}\\
 &C_3^2(\rho_{2\otimes 2\otimes 2})+\sum_{\rho_{12|34}}\frac{1}{(d_3+d_4-1)(d_1+d_2-1)}C^2(\rho_{2\otimes 2})\\
 &+\sum_{\rho_{13|24}}\frac{1}{(d_2+d_4-1)(d_1+d_3-1)}C^2(\rho_{2\otimes 2})
 +\sum_{\rho_{14|23}}\frac{1}{(d_2+d_3-1)(d_1+d_4-1)}C^2(\rho_{2\otimes 2})).\\
  \end{split}
\end{eqnarray}
\end{thm}

Next we consider the case of $N\geq5$. For any $N$-qubits $(N\geq5)$ mixed state $\rho$, the concurrence $C(\rho)$ satisfies [17]
\begin{eqnarray}
C^2(\rho)\geq\frac{N}{2^{N-2}}\sum_{i=1}^{N-1}\sum_{j>i}^NC_{ij}^2(\rho).
\end{eqnarray}
A pure $N$-partite quantum state $|\varphi\rangle_{d_1\otimes d_2\otimes\cdots\otimes d_N}\in H_1\otimes H_2\otimes \cdots\otimes H_N$ with the dimensions $d_1$, $d_2$, $\cdots$, $d_N$, respectively, has the form,
\begin{eqnarray}
|\varphi\rangle_{d_1\otimes d_2\otimes\cdots\otimes d_N}=\sum_{r_1=1}^{d_1}\sum_{r_2=1}^{d_2}\cdots\sum_{r_N=1}^{d_N}a_{r_1r_2\cdots r_N}|r_1r_2\cdots r_N\rangle,
\end{eqnarray}
where $a_{r_1r_2\cdots r_N}\in$ $\mathbb{C}$, $\sum_{r_1r_2\cdots r_N}a_{r_1r_2\cdots r_N}a_{r_1r_2\cdots r_N}^\ast=1$. Hence we get
\begin{eqnarray}
\begin{array}{ll}
 C^2(|\varphi\rangle_{d_1\otimes d_2\otimes\cdots\otimes d_N})&\\
 =\displaystyle\frac{1}{2^{N-2}}\sum_{r_1,h_1=1}^{d_1}\sum_{r_2,h_2=1}^{d_2}\cdots\sum_{r_N,h_N=1}^{d_N}(|a_{r_1r_2\cdots r_N}a_{h_1h_2\cdots h_N}-a_{h_1r_2\cdots r_N}a_{r_1h_2\cdots h_N}|^2+\cdots\\
 +|a_{r_1r_2\cdots r_N}a_{h_1h_2\cdots h_N}-a_{r_1r_2\cdots h_N}a_{h_1h_2\cdots r_N}|^2.
\end{array}
\end{eqnarray}
According to (18) and (20), using the similar method to Theorem 1 we can generalize our result to $N$-partite quantum systems as follows:

\begin{thm}
For any $d_1\otimes d_2\otimes\cdots\otimes d_N$ N-partite ($N\geq 5$) mixed state $\rho_{d_1\otimes d_2\otimes\cdots\otimes d_N}$, the concurrence $C_N(\rho_{d_1\otimes d_2\otimes\cdots\otimes d_N})$ satisfies
\begin{eqnarray}
\nonumber C_N^2(\rho_{d_1\otimes d_2\otimes\cdots\otimes d_N})\geq\frac{N}{2^{N-2}(d_1-1)(d_2-1)\cdots(d_N-1)}\sum\sum_{i=1}^{N-1}\sum_{j>i}^NC_{ij}^2(\rho_{2\otimes2\otimes\cdots\otimes2}),
\end{eqnarray}
where $\sum$ represents the sum of all possible $2\otimes2\otimes\cdots\otimes2$ mixed substates $\rho_{2\otimes2\otimes\cdots\otimes2}$.
\end{thm}

{\it Example 1}. We consider the three-qutrit state,
\begin{eqnarray}
\rho_{GGHZ}=\frac{x}{27}I_{27}+(1-x)|GGHZ\rangle\langle GGHZ|,
\end{eqnarray}
where $|GGHZ\rangle=(|000\rangle+|111\rangle+|222\rangle)/\sqrt{3}$ is a generalized $GHZ$ state and $0\leq x\leq1$.
By Theorem 1, we get $C_3(\rho)\geq\frac{3\sqrt{2}(11x-9)}{4(5x-9)}$. Fig. 1 shows that this lower bound can detect the entanglement of $|GGHZ\rangle$ for $0<x<\frac{9}{11}$.
\begin{figure}[!htb]
\centerline{\includegraphics[width=0.6\textwidth]{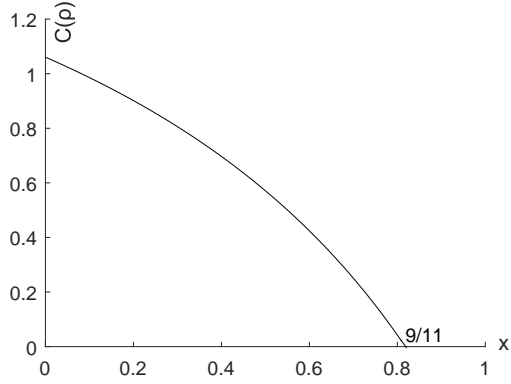}}
\renewcommand{\figurename}{Fig.}
\caption{Lower bound of $C(\rho)$ for $0\leq x\leq\frac{9}{11}$.}
\end{figure}

\section{Lower bound of concurrence for multipartite systems from qudits substates} \label{The bound}

In this section, we study lower bound of concurrence for multipartite quantum systems based on qudits substates. We focus on four-partite quantum states.
A pure four-partite quantum state $|\varphi\rangle_{d_1\otimes d_2\otimes d_3\otimes d_4}\in H_1\otimes H_2\otimes H_3\otimes H_4$ with the dimensions $d_1$, $d_2$, $d_3$ and $d_4$, respectively, has the form
\begin{eqnarray}
|\varphi\rangle_{d_1\otimes d_2\otimes d_3\otimes d_4}=\sum_{i=1}^{d_1}\sum_{j=1}^{d_2}\sum_{k=1}^{d_3}\sum_{r=1}^{d_4}a_{ijkr}|ijkr\rangle,
\end{eqnarray}
where $a_{ijkr}\in$ $\mathbb{C}$, $\sum_{ijkr}a_{ijkr}a_{ijkr}^\ast=1$. Denote $\rho_{d_1\otimes d_2\otimes d_3\otimes d_4}=|\varphi\rangle\langle\varphi|$ and
$I_0=\sum_{ijkr}a_{ijkr}a_{ijkr}^\ast$, we have
\begin{eqnarray}
I_0^2-tr(\rho_1^2)=\frac{1}{2}\sum_{i,p=1}^{d_1}\sum_{j,q=1}^{d_2}\sum_{k,t=1}^{d_3}\sum_{r,h=1}^{d_4}|a_{ijkr}a_{pqth}-a_{pjkr}a_{iqth}|^2,\nonumber\\
I_0^2-tr(\rho_2^2)=\frac{1}{2}\sum_{i,p=1}^{d_1}\sum_{j,q=1}^{d_2}\sum_{k,t=1}^{d_3}\sum_{r,h=1}^{d_4}|a_{ijkr}a_{pqth}-a_{iqkr}a_{pjth}|^2,\nonumber\\
I_0^2-tr(\rho_3^2)=\frac{1}{2}\sum_{i,p=1}^{d_1}\sum_{j,q=1}^{d_2}\sum_{k,t=1}^{d_3}\sum_{r,h=1}^{d_4}|a_{ijkr}a_{pqth}-a_{ijtr}a_{pqkh}|^2,\nonumber\\
I_0^2-tr(\rho_4^2)=\frac{1}{2}\sum_{i,p=1}^{d_1}\sum_{j,q=1}^{d_2}\sum_{k,t=1}^{d_3}\sum_{r,h=1}^{d_4}|a_{ijkr}a_{pqth}-a_{ijkh}a_{pqtr}|^2,\nonumber\\
I_0^2-tr(\rho_{12}^2)=\frac{1}{2}\sum_{i,p=1}^{d_1}\sum_{j,q=1}^{d_2}\sum_{k,t=1}^{d_3}\sum_{r,h=1}^{d_4}|a_{ijkr}a_{pqth}-a_{ijth}a_{pqkr}|^2,\nonumber\\
I_0^2-tr(\rho_{13}^2)=\frac{1}{2}\sum_{i,p=1}^{d_1}\sum_{j,q=1}^{d_2}\sum_{k,t=1}^{d_3}\sum_{r,h=1}^{d_4}|a_{ijkr}a_{pqth}-a_{pjtr}a_{iqkh}|^2,\nonumber
\end{eqnarray}
\begin{eqnarray}
I_0^2-tr(\rho_{14}^2)=\frac{1}{2}\sum_{i,p=1}^{d_1}\sum_{j,q=1}^{d_2}\sum_{k,t=1}^{d_3}\sum_{r,h=1}^{d_4}|a_{ijkr}a_{pqth}-a_{pjkh}a_{iqtr}|^2.
\end{eqnarray}
From Eq. (1), we get
\begin{eqnarray}
\begin{array}{ll}
 C^2(|\varphi\rangle_{d_1\otimes d_2\otimes d_3\otimes d_4})&\\=\frac{1}{4}\sum_{i,p=1}^{d_1}\sum_{j,q=1}^{d_2}\sum_{k,t=1}^{d_3}\sum_{r,h=1}^{d_4}(|a_{ijkr}a_{pqth}-a_{pjkr}a_{iqth}|^2+|a_{ijkr}a_{pqth}-a_{iqkr}a_{pjth}|^2 &\\
+|a_{ijkr}a_{pqth}-a_{ijtr}a_{pqkh}|^2+|a_{ijkr}a_{pqth}-a_{ijkh}a_{pqtr}|^2+|a_{ijkr}a_{pqth}-a_{pqkr}a_{ijth}|^2 &\\
+|a_{ijkr}a_{pqth}-a_{pjtr}a_{iqkh}|^2+|a_{ijkr}a_{pqth}-a_{pjkh}a_{iqtr}|^2).
\end{array}
\end{eqnarray}

For a pure state $|\varphi\rangle_{d_1\otimes d_2\otimes d_3\otimes d_4}$, its $s\otimes s\otimes s\otimes s$ pure substates $|\varphi\rangle_{s\otimes s\otimes s\otimes s}$ are of the form,
$|\varphi\rangle_{s\otimes s\otimes s\otimes s}= \sum_{i=i_1}^{i_s}\sum_{j=j_1}^{j_s}\sum_{k=k_1}^{k_s}\sum_{r=r_1}^{r_s}a_{ijkr}|ijkr\rangle=G_1\otimes G_2\otimes G_3\otimes G_4|\varphi\rangle_{d_1\otimes d_2\otimes d_3\otimes d_4}$, where $G_1=\sum_{i=i_1}^{i_s}|i\rangle\langle i|, G_2=\sum_{j=j_1}^{j_s}|j\rangle\langle j|, G_3=\sum_{k=k_1}^{k_s}|k\rangle\langle k|$ and $G_4=\sum_{r=r_1}^{r_s}|r\rangle\langle r|$, $i_s\leq d_1$, $j_s\leq d_2$, $k_s\leq d_3$ and $r_s\leq d_4$, are the projectors to $s$-dimensional subspaces, respectively.

\begin{thm}
For a four-partite mixed quantum state $\rho_{d_1\otimes d_2\otimes d_3\otimes d_4}\in H_1\otimes H_2\otimes H_3\otimes H_4$, $d_1\leq d_2\leq d_3\leq d_4$, the concurrence $C(\rho_{d_1\otimes d_2\otimes d_3\otimes d_4})$ is bounded by
\begin{eqnarray}
C^2(\rho_{d_1\otimes d_2\otimes d_3\otimes d_4})\geq\frac{1}{{d_1-2\choose s-2} {d_2-2\choose s-2} {d_3-1\choose s-1} {d_4-1\choose s-1}}\sum C^2(\rho_{s\otimes s\otimes s\otimes s}),
\end{eqnarray}
where $2\leq s\leq d_1$, ${d_i-2\choose s-2}=(d_i-2)!/((d_i-s)!(s-2)!)$, $i=1,\cdots,4$, $\sum$ stands for summing over all possible $s\otimes s\otimes s\otimes s$ mixed subtates $\rho_{s\otimes s\otimes s\otimes s}$.
\end{thm}

\begin{proof}
Consider the terms on the right hand side of Eq.(24):
\begin{eqnarray}
|a_{i_0j_0k_0r_0}a_{p_0q_0t_0h_0}-a_{p_0j_0k_0r_0}a_{i_0q_0t_0h_0}|^2,~~~ i_0\neq p_0.
\end{eqnarray}

When $j_0\neq q_0,k_0\neq t_0$ and $r_0\neq h_0$, there are 
${d_1-2\choose s-2} {d_2-2\choose s-2} {d_3-2\choose s-2} {d_4-2\choose s-2}$ different ${s\otimes s\otimes s\otimes s}$ substates. We have $|\varphi\rangle_{s\otimes s\otimes s\otimes s}=G_1\otimes G_2\otimes G_3\otimes G_4 |\varphi\rangle_{d_1\otimes d_2\otimes d_3\otimes d_4}$, where $G_1=|i_0\rangle\langle i_0|+|p_0\rangle\langle p_0|+\sum_{i=i_3}^{i_s}|i\rangle\langle i|$, $G_2=|j_0\rangle\langle j_0|+|q_0\rangle\langle q_0|+\sum_{j=j_3}^{j_s}|j\rangle\langle j|$, $G_3=|k_0\rangle\langle k_0|+|t_0\rangle\langle t_0|+\sum_{k=k_3}^{k_s}|k\rangle\langle k|$ and $G_4=|r_0\rangle\langle r_0|+|h_0\rangle\langle h_0|+\sum_{r=r_3}^{r_s}|r\rangle\langle r|$ with $i_s\leq d_1$, $j_s\leq d_2$, $k_s\leq d_3$ and $r_s\leq d_4$ respectively.

When $j_0\neq q_0,k_0= t_0$ and $r_0= h_0$, we have ${d_1-2\choose s-2} {d_2-2\choose s-2} {d_3-1\choose s-1} {d_4-1\choose s-1}$ different ${s\otimes s\otimes s\otimes s}$ substates. We have $|\varphi\rangle_{s\otimes s\otimes s\otimes s}=G_1\otimes G_2\otimes G_{3'}\otimes G_{4'}|\varphi\rangle_{d_1\otimes d_2\otimes d_3\otimes d_4}$, where $G_{3'}=|k_0\rangle\langle t_0|+\sum_{k=k_2}^{k_s}|k\rangle\langle k|$ and $G_{4'}=|r_0\rangle\langle h_0|+\sum_{r=r_2}^{r_s}|r\rangle\langle r|$.

When $j_0= q_0,k_0\neq t_0$ and $r_0= h_0$, there are ${d_1-2\choose s-2} {d_2-1\choose s-1} {d_3-2\choose s-2} {d_4-1\choose s-1}$ different ${s\otimes s\otimes s\otimes s}$ substates. We have $|\varphi\rangle_{s\otimes s\otimes s\otimes s}=G_1\otimes G_{2'}\otimes G_3\otimes G_{4'}|\varphi\rangle_{d_1\otimes d_2\otimes d_3\otimes d_4}$, where $G_{2'}=|j_0\rangle\langle q_0|+\sum_{j=j_2}^{j_s}|j\rangle\langle j|$.

When $j_0= q_0,k_0= t_0$ and $r_0\neq h_0$, there are ${d_1-2\choose s-2} {d_2-1\choose s-1} {d_3-1\choose s-1} {d_4-2\choose s-2}$ different ${s\otimes s\otimes s\otimes s}$ substates. We have $|\varphi\rangle_{s\otimes s\otimes s\otimes s}=G_1\otimes G_{2'}\otimes G_{3'}\otimes G_4|\varphi\rangle_{d_1\otimes d_2\otimes d_3\otimes d_4}$.

When $j_0\neq q_0,k_0\neq t_0$ and $r_0= h_0$, there are ${d_1-2\choose s-2} {d_2-2\choose s-2} {d_3-2\choose s-2} {d_4-1\choose s-1}$ different ${s\otimes s\otimes s\otimes s}$ substates. We have $|\varphi\rangle_{s\otimes s\otimes s\otimes s}=G_1\otimes G_2\otimes G_3\otimes G_{4'}|\varphi\rangle_{d_1\otimes d_2\otimes d_3\otimes d_4}$.

When $j_0\neq q_0,k_0= t_0$ and $r_0\neq h_o$, we have ${d_1-2\choose s-2} {d_2-2\choose s-2} {d_3-1\choose s-1} {d_4-2\choose s-2}$ different ${s\otimes s\otimes s\otimes s}$ substates. We have $|\varphi\rangle_{s\otimes s\otimes s\otimes s}=G_1\otimes G_2\otimes G_{3'}\otimes G_4|\varphi\rangle_{d_1\otimes d_2\otimes d_3\otimes d_4}$.

When $j_0= q_0,k_0\neq t_0$ and $r_0\neq h_0$, we can get ${d_1-2\choose s-2} {d_2-1\choose s-1} {d_3-2\choose s-2} {d_4-2\choose s-2}$ different ${s\otimes s\otimes s\otimes s}$ substates. We have $|\varphi\rangle_{s\otimes s\otimes s\otimes s}=G_1\otimes G_{2'}\otimes G_3\otimes G_4|\varphi\rangle_{d_1\otimes d_2\otimes d_3\otimes d_4}$.

Noting that ${d_4-2\choose s-2}\leq{d_4-1\choose s-1}$, ${d_2-2\choose s-2}\leq{d_2-1\choose s-1}$, ${d_3-1\choose s-1}{d_4-2\choose s-2}\leq{d_4-1\choose s-1}{d_3-2\choose s-2}$ and  ${d_2-1\choose s-1}{d_3-2\choose s-2}\leq{d_3-1\choose s-1}{d_2-2\choose s-2}$, we have
\begin{eqnarray}
{d_1-2\choose s-2} {d_2-2\choose s-2} {d_3-1\choose s-1} {d_4-1\choose s-1}C^2(|\varphi\rangle_{d_1\otimes d_2\otimes d_3\otimes d_4})\geq\sum C^2(|\varphi\rangle_{s\otimes s\otimes s\otimes s}).
\end{eqnarray}
Therefore we obtain that
\begin{eqnarray}
C^2(|\varphi\rangle_{d_1\otimes d_2\otimes d_3\otimes d_4})\geq \frac{1}{{d_1-2\choose s-2} {d_2-2\choose s-2} {d_3-1\choose s-1} {d_4-1\choose s-1}}\sum C^2(|\varphi\rangle_{s\otimes s\otimes s\otimes s}).
\end{eqnarray}

For the mixed state $\rho_{d_1\otimes d_2\otimes d_3\otimes d_4}$, we have
\begin{eqnarray}
\begin{split}
C(\rho_{d_1\otimes d_2\otimes d_3\otimes d_4})&=min\sum_ip_iC(|\varphi_i\rangle_{d_1\otimes d_2\otimes d_3\otimes d_4})\\
&\geq\frac{1}{\sqrt{{{d_1-2\choose s-2} {d_2-2\choose s-2} {d_3-1\choose s-1} {d_4-1\choose s-1}}}}min\sum_ip_i(\sum C^2(|\varphi_i\rangle_{s\otimes s\otimes s\otimes s\otimes s}))^\frac{1}{2}\\
&\geq\frac{1}{\sqrt{{{d_1-2\choose s-2} {d_2-2\choose s-2} {d_3-1\choose s-1} {d_4-1\choose s-1}}}}min[\sum(\sum_i p_i C(|\varphi_i\rangle_{s\otimes s\otimes s\otimes s\otimes s}))^2]^\frac{1}{2}\\
&\geq\frac{1}{\sqrt{{{d_1-2\choose s-2} {d_2-2\choose s-2} {d_3-1\choose s-1} {d_4-1\choose s-1}}}}[\sum(min\sum_i p_i C(|\varphi_i\rangle_{s\otimes s\otimes s\otimes s\otimes s}))^2]^\frac{1}{2}\\
&=\frac{1}{\sqrt{{{d_1-2\choose s-2} {d_2-2\choose s-2} {d_3-1\choose s-1} {d_4-1\choose s-1}}}}[\sum C^2(\rho_{s\otimes s\otimes s\otimes s\otimes s})]^\frac{1}{2},\\
\end{split}
\end{eqnarray}
where the minimum is taken over all possible pure state decompositions of the mixed state $\rho_{d_1\otimes d_2\otimes d_3\otimes d_4}$ in the first three minimizations, the minimum in the last inequality is taken over all pure state decompositions of $\rho_{s\otimes s\otimes s\otimes s}$, the first three $\sum$ stand for the summation over all possible $s\otimes s\otimes s\otimes s$ mixed subtates $|\varphi_i\rangle_{s\otimes s\otimes s\otimes s}$  and the last $\sum$ stands for summing over all possible $s\otimes s\otimes s\otimes s$ pure subtates $\rho_{s\otimes s\otimes s\otimes s}$. The Minkowski inequality $(\sum_j(\sum_ix_{ij})^2)^\frac{1}{2}\leq\sum_i(\sum_j x_{ij}^2)^\frac{1}{2}$ has been used in the second inequality.
\end{proof}
{\bf Remark 2.} In Theorem 4, we derive the lower bound of concurrence for four-partite quantum systems, thus Theorem 4 is a generalization of the Theorem 1 given in [18]. Theorem 2 in [14] gives the lower bound of concurrence of $m\otimes n\otimes l$ tripartite quantum states, i.e, $C^2(\rho)\geq[{m-2\choose s-2}{n-2\choose s-2}{l-1\choose s-1}]^{-1}\sum C^2(\rho_{s\otimes s\otimes s})$. Thus, Theorem 4 is also a generalization of the Theorem 2 in [14].

{\it Example 2}. Let us consider the $2\otimes2\otimes2\otimes3$ state
\begin{eqnarray}
\rho_{1234}=\frac{1-x}{16}I_{16}+x|\psi\rangle\langle \psi|,
\end{eqnarray}
where $|\psi\rangle=\frac{1}{2}(|0000\rangle+|0012\rangle+|1100\rangle+|1112\rangle)$ and $0\leq x\leq1$. By Theorem 4, the lower bound of concurrence is
$C(\rho_{1234})\geq\ \frac{\sqrt{7}}{4}\frac{\sqrt{x+2x\sqrt{3x+1}+4x^2+2}-\sqrt{x-2x\sqrt{3x+1}+4x^2+2}}{5+3x}$. From the lower bound in [23], one has $C(\rho_{1234})\geq\frac{3x-1}{2}$.
From Fig. 1, our bound is better than that of [23] for $\frac{1}{3}\leq x \leq 0.4$, showing that our bound from Theorem 4 provides a better estimation of concurrence than that of [23].

\begin{figure}[!htb]
\centerline{\includegraphics[width=0.6\textwidth]{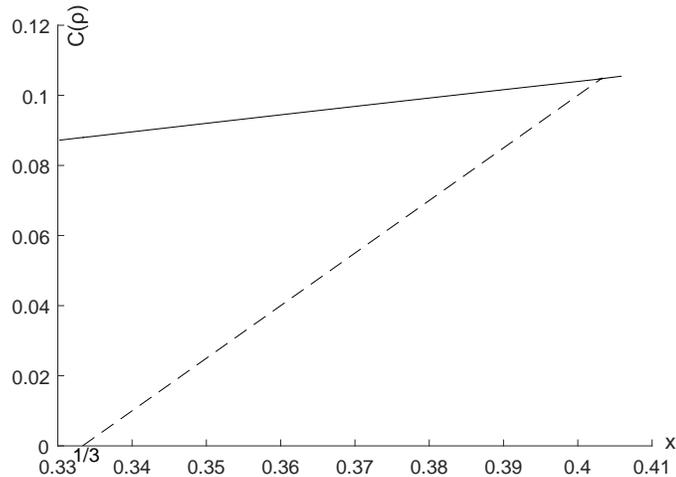}}
\renewcommand{\figurename}{Fig.}
\caption{Lower bounds of concurrence from Theorem 4 (solid curve) and from [23] (dashed line).}
\end{figure}

Choosing different subspace dimensions $s$ may give rise to different lower bounds. We can obtain a new lower bound by convex combination of these lower bounds.

{\noindent\bf Corollary 1} For a four-partite mixed quantum state $\rho_{d_1\otimes d_2\otimes d_3\otimes d_4}\in H_1\otimes H_2\otimes H_3\otimes H_4$, $d_1\leq d_2\leq d_3\leq d_4$, the concurrence is bounded by
\begin{eqnarray}
C^2(\rho_{d_1\otimes d_2\otimes d_3\otimes d_4})\geq\sum \sum_{s=2}^m \frac{p_s}{{d_1-2\choose s-2} {d_2-2\choose s-2} {d_3-1\choose s-1} {d_4-1\choose s-1}}C^2(\rho_{s\otimes s\otimes s\otimes s}),
\end{eqnarray}
where $0\leq p_s\leq1,s=2,\cdots,m$, $\sum_{s=2}^mp_s=1$, ${d_i-2\choose s-2}=(d_i-2)!/((d_i-s)!(s-2)!)$, $i=1,\cdots,4$, and $\sum$
sums over 
all possible $s\otimes s\otimes s\otimes s$ mixed subtates $\rho_{s\otimes s\otimes s\otimes s}$.

{\bf Remark 3.} Our above approach can be generalized to multipartite quantum systems, by taking into account the terms on the right hand side of Eq. (20),
$|a_{{r_1}'{r_2}'\cdots {r_N}'}a_{{h_1}'{h_2}'\cdots {h_N}'}-a_{{h_1}'{r_2}'\cdots {r_N}'}a_{{r_1}'{h_2}'\cdots {h_N}'}|^2$, ${r_1}'\neq {h_1}'$.
Similar analysis provides lower bounds of concurrence for multipartite quantum states.

\section{Conclusion}

We derived lower bounds of concurrence for tripartite mixed quantum states $\rho_{d_1\otimes d_2\otimes d_3}$ by projecting to three-qubit quantum states using the monogamy property of concurrence.
The results are generalized to multipartite quantum systems.
Moreover, by analyzing the concurrence of a pure four-partite quantum state $|\varphi\rangle_{d_1\otimes d_2\otimes d_3\otimes d_4}$,
we have projected a high-dimensional four-partite quantum state to lower $s$-dimensional systems, and lower bounds of concurrence for any four-partite quantum mixed states
are obtained. By detailed examples we have shown that these bounds are better than other bounds given in the literature.

\textbf {Acknowledgements}

This work is supported by the National Natural Science Foundation of China under grant Nos. 11101017, 11531004, 11726016 and 11675113,
and Simons Foundation under grant No. 523868, Key Project of Beijing Municipal Commission of Education (KZ201810028042).

\end{document}